\documentclass{article}
\usepackage{spconf,amsmath,graphicx,amssymb,bm,amsthm,subcaption,dsfont}
\usepackage{algorithmic,algorithm}
\usepackage{balance,cite,color}

\newcommand{\grad}[1]{\nabla_{#1^{\tran}} }
\newcommand{\w}{\bm{w}}

\newcommand{\we}{\widetilde{\w}}
\newcommand{\eqdef}{\:\overset{\Delta}{=}\:}
\DeclareMathOperator*{\argmin}{argmin}

\newcommand{\tran}{{\sf T}}

\newtheorem{theorem}{Theorem}
\newtheorem{assumption}{Assumption}

\newtheorem{definition}{Definition}

\title{Local Graph-homomorphic Processing for Privatized Distributed Systems}
\name{Elsa Rizk\thanks{School of Engineering, École Polytechnique Fédérale de Lausanne (e-mail:\{elsa.rizk, ali.sayed\}@epfl.ch).}, Stefan Vlaski\thanks{Department of Electrical and Electronic Engineering, Imperial College London  (e-mail: s.vlaski@imperial.ac.uk).}, Ali H. Sayed}
\address{}

\begin{document}
\ninept
\maketitle
\begin{abstract}
We study the generation of dependent random numbers in a distributed fashion in order to enable privatized distributed learning by networked agents. We propose a method that we refer to as local graph-homomorphic processing; it relies on the construction of particular noises over the edges to ensure a certain level of differential privacy. We show that the added noise does not affect the performance of the learned model. This is a significant improvement to previous works on differential privacy for distributed algorithms, where the noise was added in a less structured manner without respecting the graph topology and has often led to performance deterioration. We illustrate the theoretical results by considering a linear regression problem over a network of agents.  
\end{abstract}
\section{Introduction and Related Material}
\label{sec:intro}
Distributed systems consist of a network of agents that collaborate to achieve a common goal. Some examples include distributed computing \cite{Apt2009} when components of a software are shared over a network, or distributed machine learning \cite{verbraeken2020survey} where the goal is to fit a global model to the data dispersed at different computing locations. During collaboration in such systems, communication between neighbours is necessary. However, the shared information might be sensitive, such as in distributed systems handling health or financial data. Thus, there is a need to privatize communication channels. One way to achieve secure communications is through cryptographic methods \cite{bonawitz2016practical,Mohassel2017SecureML,froelicher2021scalable,Niko2013}, while another is by adding random noise to make the communication differentially private \cite{dwork2014algorithmic, JayDLDP,LiDLDP,pathak2010DP,vlaski2020graphhomomorphic}. 
	
In the standard implementations, agents add \textit{independent} noise to their shared messages. This property degrades the performance of the learned model since the noises propagate over the graph through cooperation, as already shown in Theorem 1 of \cite{GFL}. In order to endow agents with enhanced privacy with minimal effect on  performance, it is necessary for the additional noise sources to be mindful of the graph topology \cite{vlaski2020graphhomomorphic}. However, this information is not available globally and, therefore, one needs to devise a scheme to generate graph-dependent random noise sources in a distributed manner and without assuming any global information about the graph structure. Motivated by this observation, we develop in this work a scheme that constructs privacy perturbations in a manner that their negative effect on performance is canceled out. One solution was suggested in \cite{bonawitz2016practical} for the case of federated learning. Pairs of agent collaborate to add noise that cancels out at the server. However, the suggested method generates pseudo-random numbers, which is less secure than true random numbers \cite{10.1145/3243734.3243756} and without any guarantees of differential privacy. 

The objective of this work is therefore to generate dependent random numbers in a distributed manner across a graph. The problem is challenging for at least two reasons. Firstly, generating random numbers is usually difficult without enforcing beforehand some distribution for the random process. In practice, random number generators exploit a variety of entropy sources in a computer such as mouse movements, available memory, or temperature \cite{johnston2018random}. Secondly, it is not evident how agents should exploit independent entropy sources to generate \textit{dependent} random numbers. Most available solutions \cite{8907316,cascudo2017scrape,syta2017scalable,hanke2018dfinity,schindler2018hydrand,popov2017decentralized,Blu} rely on a central orchestrator or consider a fully connected network. A truly distributed method does not appear to exist.

\section{Local Graph-homomorphic Process}\label{sec:LGH}

\subsection{Problem Setup}
We consider a network of $K$ agents connected by some graph topology (Fig. \ref{fig:net}). We let $a_{mk} > 0$ denote the weight attributed to the message sent by neighbour $m$ to agent $k$ and let $A=[a_{mk}]$ denote  the corresponding combination matrix. We assume $A$ is symmetric and doubly-stochastic, i.e.:
\begin{equation}
	\mathds{1}^\tran A = \mathds{1}^\tran \quad A\mathds{1}  = \mathds{1}.
\end{equation}
We further denote the neighbourhood of agent $k$ by $\mathcal{N}_k$; it consists of all agents connected to $k$ by an edge. 
\begin{figure}[h!]
	\centering
	\includegraphics[width=0.4\textwidth]{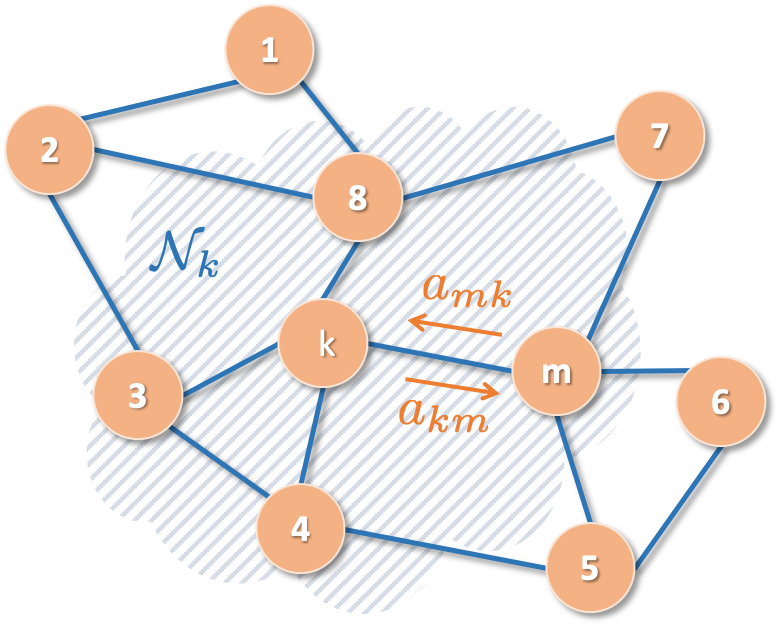}
	\caption{Illustration of a network of agents.}\label{fig:net}
\end{figure}

We consider problems where agents aggregate the received messages from their neighbours. In other words, if we let $\bm{\psi}_{mk,i}$ denote the message sent by agent $m$ to agent $k$ at time $i$, then: 
\begin{align}\label{eq:aggr}
	\w_{k,i} = \sum_{m \in \mathcal{N}_k}a_{mk} \bm{\psi}_{mk,i},
\end{align}
which is the aggregate of all messages arriving at agent $k$. We wish to secure the communication between the agents. One method is to mask the messages with some random noise to guarantee some level of differential privacy. If we denote by $\bm{g}_{mk,i}$ the noise added to the message $\bm{\psi}_{mk,i}$, then the secure aggregation becomes:
\begin{align}\label{eq:secAggr}
	\w_{k,i} = \sum_{m \in \mathcal{N}_k} a_{mk} \left(\bm{\psi}_{mk,i} + \bm{g}_{mk,i} \right).
\end{align}

Ideally, we would like that no information is lost by the added noise and that the aggregate message is equivalent to the non-noisy version. This is guaranteed if the noise sources added in \eqref{eq:secAggr} satisfy the following condition in every neighbourhood:
\begin{align}\label{eq:condLGH}
	\sum_{m \in \mathcal{N}_k} a_{mk} \bm{g}_{mk,i} = 0.
\end{align}
Noises that satisfy \eqref{eq:condLGH} are said to arise from a {\em local graph-homomorphic process}. This is in contrast to the {\em global graph-homomorphic process} proposed in \cite{vlaski2020graphhomomorphic} where condition \eqref{eq:condLGH} is replaced by one that should hold over the entire graph, namely: 
\begin{align}
	\sum_{k=1}^K \sum_{m\in \mathcal{N}_k} a_{mk}\bm{g}_{mk,i} = 0.
\end{align}

We would also like the  noises $\bm{g}_{mk,i}$ added in \eqref{eq:secAggr} to ensure some level of differential privacy. This means that if the agent $m$ chooses to share different messages $\bm{\psi}'_{mk,i}$, then an observer would be oblivious to this change. This is more formally defined as follows. 

\begin{definition}[\textbf{$\epsilon(i)-$Differential Privacy}]\label{def:DP}
	We say the communication is $\epsilon(i)-$differentially private for agent $m$ at time $i$ if the following condition on the probability of observing the respective events holds for all agents:
	\begin{align}\label{eq:condDP}
		\frac{\mathbb{P} \left( \left\{  \left\{ \bm{\psi}_{mk,j} + \bm{g}_{mk,j} \right\}_{k\in \mathcal{N}_m \setminus \{m\}}  \right\}_{j=0}^i  \right)}{\mathbb{P} \left( \left\{  \left\{ \bm{\psi}'_{mk,j} + \bm{g}_{mk,j}  \right\}_{k\in \mathcal{N}_m\setminus \{m\}}  \right\}_{j=0}^i  \right)} \leq e^{\epsilon(i)}.
	\end{align}
\qed
\end{definition}

\subsection{Process Description}
To motivate the local graph-homomorphic process, we examine the following example. Alice and Bob wish to communicate to Charlie the aggregate of their messages without Charlie knowing the individual messages. Alice and Bob decide to send a noisy version of their messages to Charlie. However, they wish when their noisy messages are aggregated by Charlie that he will still be able to retrieve the original sum. One way to do so is by ensuring that the noises generated by Alice and Bob cancel out when Charlie computes a weighted sum of the messages. For example, they could agree on some random number $\bm{x}$ that Alice would add to her message while Bob would subtract it from his message. Now assume that all communications between Alice and Bob need to go through Charlie, i.e., no direct communication channel exists between Alice and Bob. Then, in this case, both Alice and Bob will need to agree on the random number $\bm{x}$ without explicitly mentioning it. In other words, secure communication between them will need to be set up through Charlie. One way of doing so is through the Diffie-Helman key exchange protocol \cite{diffieHell}. 

Let Alice and Bob have individual secret keys $\bm{v}_1$ and $\bm{v}_2$, respectively.  Let $p$ be a known prime number and $b$ a base. Then, both Alice and Bob will broadcast their public keys $\bm{V}_1 = b^{\bm{v}_1} \mod p$ and $\bm{V}_2 = b^{\bm{v}_2} \mod p$. When they raise the public key of the other by their secret key and take the modulus $p$, they will now share a common secret  key $\bm{v}_{12} = b^{\bm{v}_1\bm{v}_2} \mod p$. This secret key can be used as the added noise; while Alice adds $\bm{v}_{12}$ to her message, Bob can subtract it. However, to ensure the communication is differentially private, one choice of distribution of the noise is the Laplace distribution Lap$(0,\sigma_g/\sqrt{2})$. A Laplace random variable can be generated from two uniform random variables by taking the log of the ratio of the two variables and then multiplying by the inverse of the scale parameter, namely, $\sqrt{2}/\sigma_g$. Thus, to generate a Laplace random variable, we require two secret keys $\{\bm{v}_{12}, \bm{v}'_{12}\}$ that are uniformly distributed. For $\bm{v}_{12}$ to be a uniform random variable, one of the local secret keys must be uniformly distributed over $[0,1]$ while the other must be sampled from a gamma distribution $\Gamma(2,1)$. Furthermore, the base must be set to $b = e^{-1}$ and then scaled by some constant $a$ that is a multiple of the prime number $p$. Therefore, for instance, Alice should sample two uniformly distributed secret keys $\{\bm{v}_1, \bm{v}'_1 \}\sim U ([0,1])$, and Bob must generate two secret keys $\{\bm{v}_2, \bm{v}'_2\}$ from a gamma distribution. The resulting two shared secret keys will be uniformly distributed on $[0,p]$. Then, setting:
\begin{equation}
	\bm{x} = \frac{\sqrt{2}}{\sigma_g} \ln \left( \frac{\bm{v}_{12}}{\bm{v}'_{12}} \right),
\end{equation}	
results in a Laplace noise, which Alice can add to her message while Bob subtracts it from his.

Returning to the network setting, we describe the process by which the agents generate their local graph-homomorphic noises. Each agent randomly splits its neighbourhood into two groups, $\mathcal{N}_k = \mathcal{N}_{+}\bigcup \mathcal{N}_{-}$, and communicates the split to its neighbourhood. One method of splitting the neighbourhood is by attributing to each neighbour a number, and then placing all the even-numbered agents in one set, and the odd-numbered agents in the other set. Then, every pair of agents from the two sub-neighbourhoods will generate together a shared noise, with the agent in $\mathcal{N}_{+}$ adding the noise to its message and the agent in $\mathcal{N}_{-}$ subtracting it. The communication betwen the agents of the sub-neighbourhoods occurs through the main agent $k$, since these agents might not be neighbours (e.g., agents 4 and $m$ in Fig. \ref{fig:net}). The messages are scaled by the weights attributed to the neighbours by a given agent. Thus, we force each neighbour to scale its noise by the inverse of the attributed weight. For agents $\ell \in \mathcal{N}_+$ and $m \in \mathcal{N}_{-}$, we denote the generated noise by $\bm{g}_{\{\ell m\}k,i}$ where we now add the subscript $\{\ell m\}$ to indicate that the noise was generated by the pair of agents. We follow the convention of writing the subscript of the agent from the positive set first, followed by that from the negative set. Then, every neighbour $\ell$ will send agent $k$ its message masked by the sum of all the noise it generated with the agents from the adjacent sub-neighbourhood. 
A more detailed description of the process is found in Algorithm \ref{alg:process}. An illustrative example is found in  Fig. \ref{fig:locGHP}.

\begin{algorithm}[h!]
	\begin{algorithmic}
		\caption{(Local graph-homomorphic processing)}\label{alg:process}
		\FOR{each iteration $i=1,2,\cdots$} \STATE{
			\FOR{each agent $k=1,2,\cdots, K$} \STATE {
				Split $\mathcal{N}_k = \mathcal{N}_+ \bigcup \mathcal{N}_-$ and communicate the split to the neighbours.
				\FOR{each pair of agents $\ell \in \mathcal{N}_+ $ and $m \in \mathcal{N}_-$}\STATE{
					Agent $\ell$ samples two secret keys $\{\bm{v}_{\ell}, \bm{v}'_{\ell}\} \sim U ([0,1])$, and angent $m$ samples two secret keys $\{\bm{v}_{m}, \bm{v}'_{m}\} \sim \Gamma (2,1)$. 
					\\
					Calculate and broadcast the public keys
					\begin{align*}
						\bm{V}_{\ell} &= a e^{-v_{\ell}} \mod p 
						\\
						\bm{V}'_{\ell} &= a e^{-v'_{\ell}} \mod p
						\\
						\bm{V}_{m} &= a e^{-v_{m}} \mod p 
						\\
						\bm{V}'_{m} &= a e^{-v'_{m}} \mod p
					\end{align*}
					\\
					Calculate the shared secret keys
					\begin{align*}
						\bm{v}_{\ell m} &= a e^{-\bm{v}_{\ell} \bm{v}_m} \mod p
						\\
						\bm{v}'_{\ell m} &= a e^{-\bm{v}'_{\ell} \bm{v}'_m} \mod p
					\end{align*}
					\\
					Set the noise 
					\begin{align*}
						\bm{g}_{\{\ell m\}k,i} = \frac{\sqrt{2}}{\sigma_g} \ln \left( \frac{\bm{v}_{\ell m}}{\bm{v}_{\ell m}'} \right)
					\end{align*}
					\\
				}\ENDFOR
			\FOR{each agent $\ell \in \mathcal{N}_k $ }\STATE{ 
				\IF{$\ell \in \mathcal{N}_+$}
					\STATE {Send $\bm{\psi}_{\ell k,i} + \sum\limits_{m \in \mathcal{N}_-} \bm{g}_{\{\ell m\}k,i}/a_{\ell k}$}
				\ELSE
				\STATE {Send $\bm{\psi}_{\ell k,i} - \sum\limits_{m \in \mathcal{N}_+} \bm{g}_{\{m \ell\}k,i}/a_{\ell k}$}
					\ENDIF
		} \ENDFOR	
			} \ENDFOR
		}\ENDFOR
		
	\end{algorithmic}
\end{algorithm}

\begin{figure}[h!]
	\centering
	\includegraphics[width =0.5\textwidth]{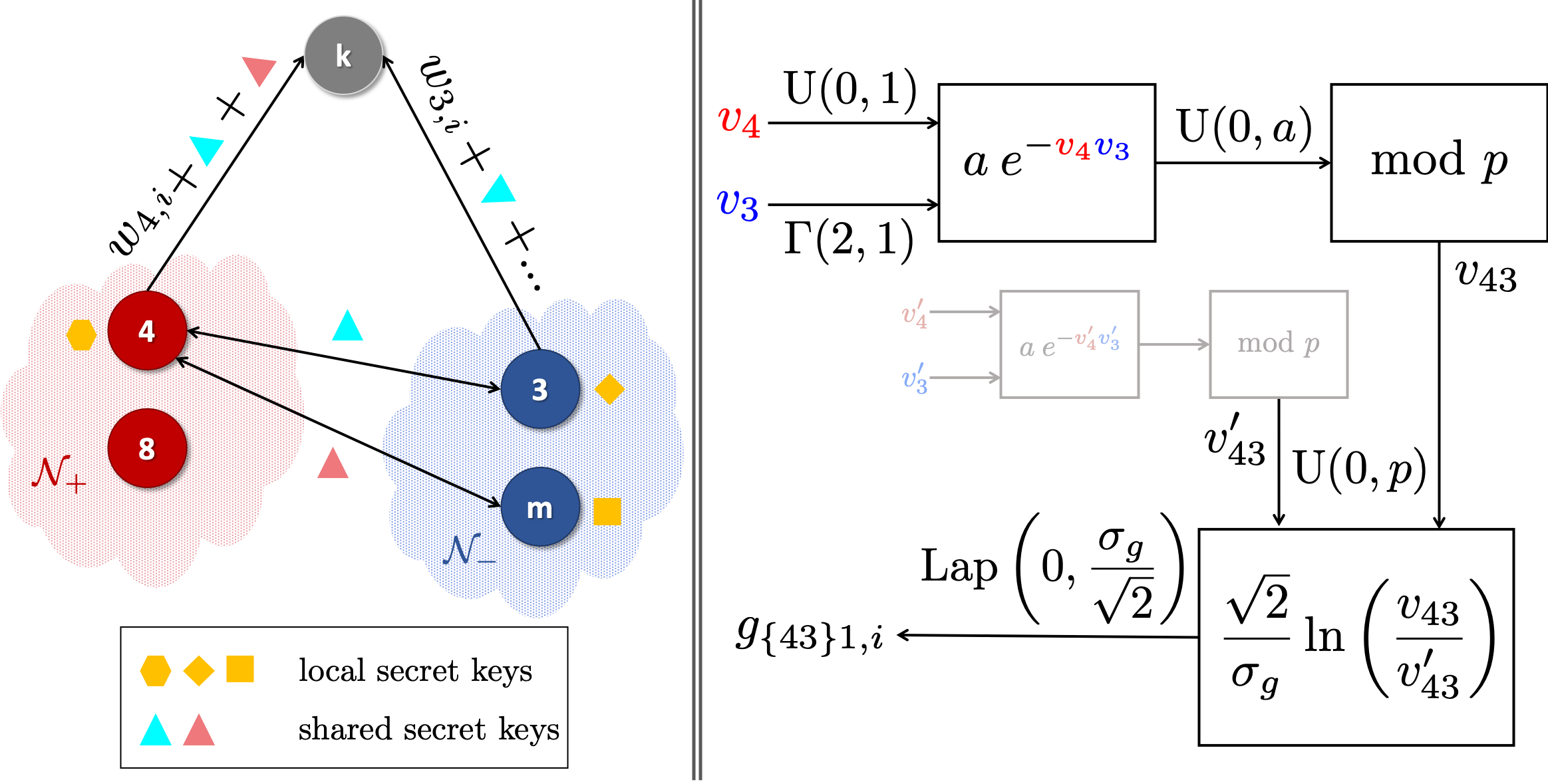}
	\caption{Illustration of the local graph-homomorphic process with the Diffie-Helman key exchange protocol on the left and the transformation of the random variable  on the right .}\label{fig:locGHP}
\end{figure}


\section{ Privatized Distributed Learning}\label{sec:distLearn}
We apply the above construction to the problem where a network of agents aims to solve an aggregate convex optimization problem of the form:
\begin{align}
	w^o \eqdef \argmin_{w\in\mathbb{R}^M} \frac{1}{K}\sum_{k=1}^K \left\{J_k(w) \eqdef \frac{1}{N_k}\sum_{n=1}^{N_k} Q_{k}(w;x_{k,n}) \right\},
\end{align}  
where the risk function $J_k(\cdot)$ is associated with agent $k$ and is defined as an empirical average of the loss function $Q_k(\cdot;\cdot)$ evaluated over the local dataset $\{x_{k,n}\}_{n=1}^{N_k}$. We assume the loss functions are convex with Lipschitz continuous gradients and the risk functions are strongly convex. 

\begin{assumption}[\textbf{Convexity and smoothness}]\label{ass:fct}
	The empirical risks $J_{k}(\cdot)$ are $\nu-$strongly convex, and the loss functions $Q_{k}(\cdot;\cdot)$ are convex and twice differentiable, namely for some $\nu > 0$:
	\begin{align} \label{eq:assFctConv}
		&J_{k}(w_2) \geq  \: J_{k}(w_1) + \grad{w}J_{k}(w_1)(w_2-w_1) + \frac{\nu}{2}\Vert w_2 - w_1 \Vert^2,\\
		&Q_{k}(w_2;\cdot) \geq  \: Q_{k}(w_1;\cdot) + \grad{w}Q_{k}(w_1;\cdot) (w_2 - w_1).
	\end{align}
	Furthermore, the loss functions have $\delta-$Lipschitz continuous gradients:
	\begin{equation}\label{eq:assFctLip}
		\Vert \grad{w}Q_{k}(w_2;x_{k,n}) - \grad{w}Q_{k}(w_1;x_{k,n})\Vert \leq \delta \Vert w_2 - w_1\Vert.
	\end{equation}
	\qed
\end{assumption}
We next make an assumption on the drift between the local optimal models $w^o_k = \argmin J_k(w)$ and the global optimal model $w^o$. For collaboration to make sense, the drift must be bounded. In case the difference is not bounded, then the agents should not collaborate to find one global model since that global model will not perform well locally. 
\begin{assumption}[\textbf{Model drifts}]\label{ass:mod}
	The distance of each local model $w_k^o$ to the global model $w^o$ is uniformly bounded, $\Vert w^o - w_k^o\Vert \leq \xi$. 
	\qed
\end{assumption}

To approximate the optimal model $w^o$, the agents can collaborate and run a distributed algorithm like consensus \cite{DeGroot,Johansson2008SubgradientMA,Nedic09} or diffusion \cite{chen2012diffusion,Tu12}, while at the same time adding noise to their messages to ensure a certain level of privacy. For instance, the privatized adapt-then-combine (ATC) diffusion algorithm would take the following form:
\begin{align}
	\bm{\psi}_{k,i} &= \w_{k,i-1} - \mu \grad{w}Q_{k} (\w_{k,i-1};\bm{x}_{k,b}),  \label{eq:ATC-adapt}
	\\
	\w_{k,i} &= \sum_{m\in \mathcal{N}_k} a_{mk} \left( \bm{\psi}_{m,i} + \bm{g}_{mk,i}\right) \label{eq:ATC-combine},
\end{align}
where we now drop the second subscript $k$ from the message $\bm{\psi}_{m,i}$ since the same message is sent to all the neighbours of agent $m$, i.e., $\bm{\psi}_{mk,i} = \bm{\psi}_{m,i}$. 
Then, because $\bm{g}_{mk,i}$ is sampled from a Laplacian distribution, this construction ensures that  the algorithm is $\epsilon(i)-$differentially private for some choice of variance $\sigma_g^2$ (see Theorem 2 in  \cite{GFL}). Recal that the local graph-homomorphic noises in \eqref{eq:ATC-combine} are generated from the Laplacian noises $\bm{g}_{mk,i}$:
\begin{equation}
	\bm{g}_{mk,i} = \begin{cases}
		-\sum\limits_{\ell \in \mathcal{N}_+} \bm{g}_{\{\ell m\}k,i}, & m \in \mathcal{N}_- \\
		\sum\limits_{\ell \in \mathcal{N}_-} \bm{g}_{\{m \ell\}k,i}. & m \in \mathcal{N}_+
	\end{cases}
\end{equation}
Since, by construction, the noises cancel out, the performance of the privatized ATC diffusion strategy \eqref{eq:ATC-adapt}--\eqref{eq:ATC-combine} ends up being equivalent to the performance of the traditional non-privatized strategy without degradation. Thus, the algorithm will still converge to an $O(\mu)$ neighbourhood of the optimal model $w^o$. This is a significant improvement compared to earlier results where the limiting neighborhood was on the order of $O(\mu^{-1})$ or $O(1)$ --- see, e.g., \cite{GFL,vlaski2020graphhomomorphic} and the discussions therein.

\begin{theorem}[\textbf{MSE convergence}]\label{thrm:MSE}
	Under assumptions \ref{ass:fct} and \ref{ass:mod}, the privatized diffusion strategy  \eqref{eq:ATC-adapt}$-$\eqref{eq:ATC-combine} with noise generated from the local graph-homomorphic process described earlier, converges exponentially fast for a small enough step-size to a neighbourhood of the optimal model:
	\begin{equation}
		\limsup_{i\to \infty} \mathbb{E}\Vert \we_{i}\Vert^2 \leq v_2^2 \mathds{1}^\tran (I - \Gamma)^{-1} \begin{bmatrix}
			v_1^2 \mu^2 \sigma_s^2 \\
			v_1^2 \mu^2 \sigma_s^2 + \frac{3\Vert \check{b}\Vert^2}{1-\rho(J_{\epsilon}) - \epsilon}
		\end{bmatrix},
	\end{equation}
for some constants $v_1^2, v_2^2, \epsilon, \rho(J_{\epsilon}), \check{b}$, $\sigma_s^2$ the bound on the variance of the gradient noise, and the convergence rate:
\begin{align}
	\Gamma \eqdef \begin{bmatrix}
		1-O(\mu) + O(\mu^2) & O(\mu^{0.5}) \\
		O(\mu) & \sqrt{\rho(J_{\epsilon}) + \epsilon} +O(\mu^2)
	\end{bmatrix}.
\end{align} 
\end{theorem}  

\begin{proof}
	Since the noise cancels out locally during each iteration, the algorithm is equivalent to the non-privatized version. The proof then follows the arguments used to establish Theorem 9.1 in \cite{sayed2014adaptation}.
\end{proof}

In the next theorem, we explain that the proposed algorithm is differentially private. 

\begin{theorem}[\textbf{Privacy of distributed learning}]\label{thrm:DP}
	Under the local graph-homomorphic process, the privatized diffusion algorithm \eqref{eq:ATC-adapt}--\eqref{eq:ATC-combine} is $\epsilon(i)-$differentially private with:
	\begin{align}
		\epsilon(i) \eqdef &\frac{2\sqrt{2}}{\sigma_g} \Bigg\{  \left( \frac{1- (1-O(\mu))^{i+1}}{O(\mu)} - 1 \right)a + b 
		+ O(\mu^{0.5}) (i-1)   \Bigg\},
	\end{align}
where $a$ and $b$ are some constants. 
\end{theorem}
\begin{proof}
	We provide a sketch of the proof. We first show that the generated noise from the local graph-homomorphic process is Laplacian. Then, using a bound on the gradients at each step of the algorithm, we can bound the sensitivity of the algorithm. This can then be used to establish condition \eqref{eq:condDP} in Definition \ref{def:DP}.
\end{proof}

As time passes, $\epsilon(i)$ increases which means higher privacy loss. To mitigate this problem, the noise variance can be increased to guarantee a certain level of privacy. Since the variance of the perturbations does not affect the MSE bound, we do not hinder the model utility by increasing the variance, as opposed to the traditional differentially privatized algorithms (where the noises are not graph-homomorphic); in these cases, the MSE will worsen by an $O(\mu^{-1})\sigma_g^2$ factor.

\section{Experimental Results}\label{sec:exp}
We study a linear regression problem over a network of $K=30$ agents with a regularized quadratic loss:
\begin{align}
	\min_{w \in \mathbb{R}^2} \frac{1}{30\times 100}\sum_{k=1}^{30}\sum_{n=1}^{100} \Vert \bm{d}_k(n) - \bm{u}_{k,n}^\tran w \Vert^2 + 0.01\Vert w\Vert^2. 
\end{align}
We generate for each agent 100 data samples $\{ \bm{u}_{k,n}, \bm{d}_k(n)\}$. We sample two-dimensional feature vectors $\bm{u}_{k,n} \sim \mathcal{N}(0,R_u)$ and an independent noise $\bm{v}_p(n) \sim \mathcal{N}(0,\sigma_{v,k}^2)$ such that $\bm{d}_k(n) = \bm{u}_{k,n}^\tran w^{\star} + \bm{v}_k(n)$ for some generative model $w^\star $. The optimal model is given by:
\begin{align}
	w^o = (\widehat{R}_u + 0.01I)^{-1} (\widehat{R}_u w^{\star} + \widehat{r}_{uv}),
\end{align}
where $\widehat{R}_u $ and $ \widehat{r}_{uv}$ are the respective sample covariance matrix and cross-covariance.

We set the step-size $\mu = 0.4$, the noise variance $\sigma_g^2 = 0.01$, and the total number of iterations 1000. We repeat the algorithm 20 times and calculate the average MSD of the centroid model defined as:
\begin{equation}
	\w_{c,i} \eqdef \frac{1}{K}\sum_{k=1}^K \w_{k,i},
\end{equation}
and the individual models:
\begin{align}
	\text{MSD}_i& \eqdef \Vert \w_{c,i} - w^o\Vert^2, \\
	\text{MSD}_{\text{avg},i} & \eqdef \frac{1}{K}\sum_{k=1}^K \Vert \w_{k,i} - w^o\Vert^2.
\end{align}
We plot the results of the non-privatized algorithm, the privatized algorithm with random perturbations, and the privatized algorithm with local graph-homomorphic perturbations. As expected from Theorem \ref{thrm:MSE}, the local graph-homomorphic perturbations do not affect the performance of the algorithm.

\begin{figure}[h!]
	\centering
	\begin{subfigure}[b]{0.4\textwidth}
		\centering
		\includegraphics[width=\textwidth]{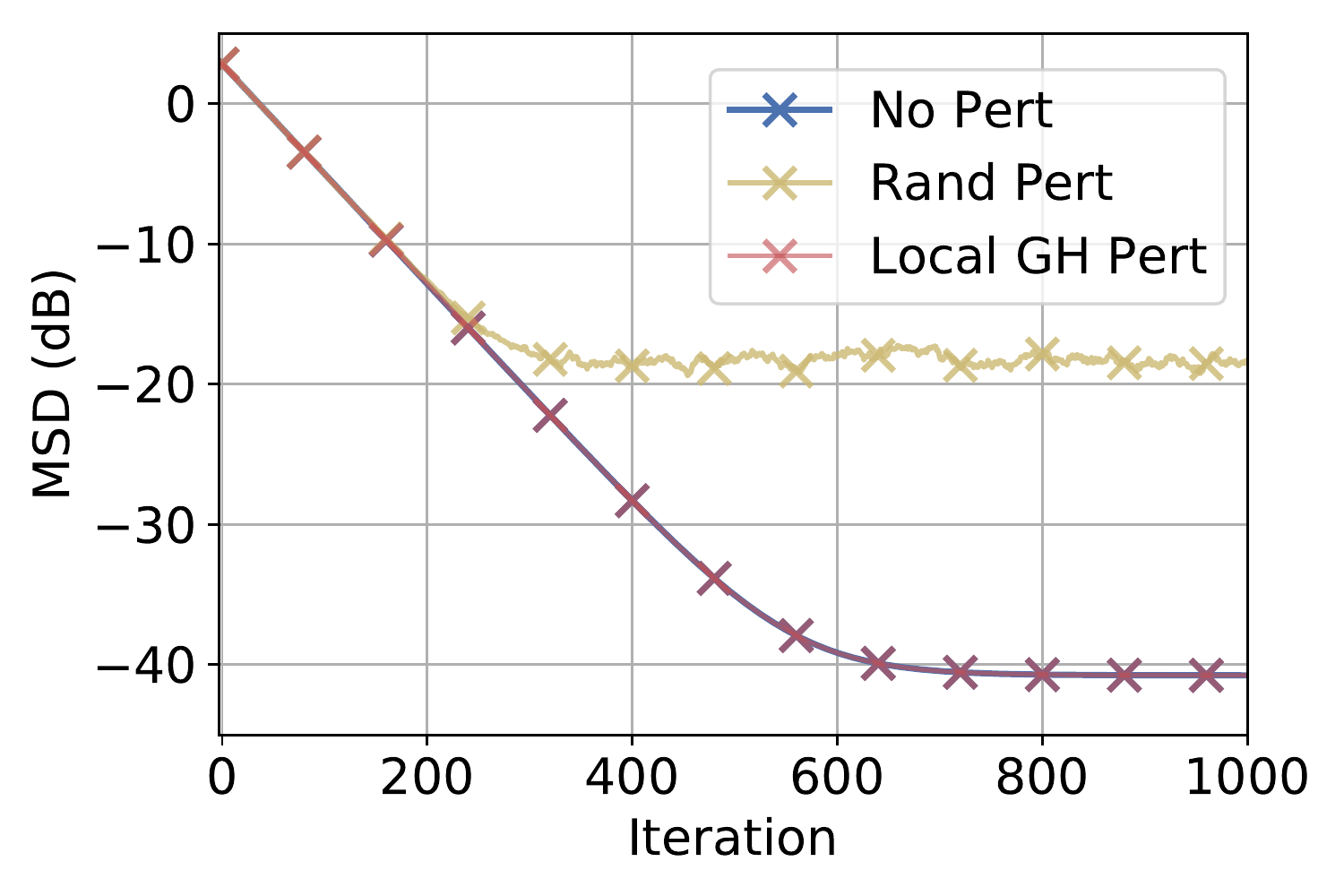}
		\caption{Centroid MSD}
	\end{subfigure}
	\begin{subfigure}[b]{0.4\textwidth}
		\centering
		\includegraphics[width=\textwidth]{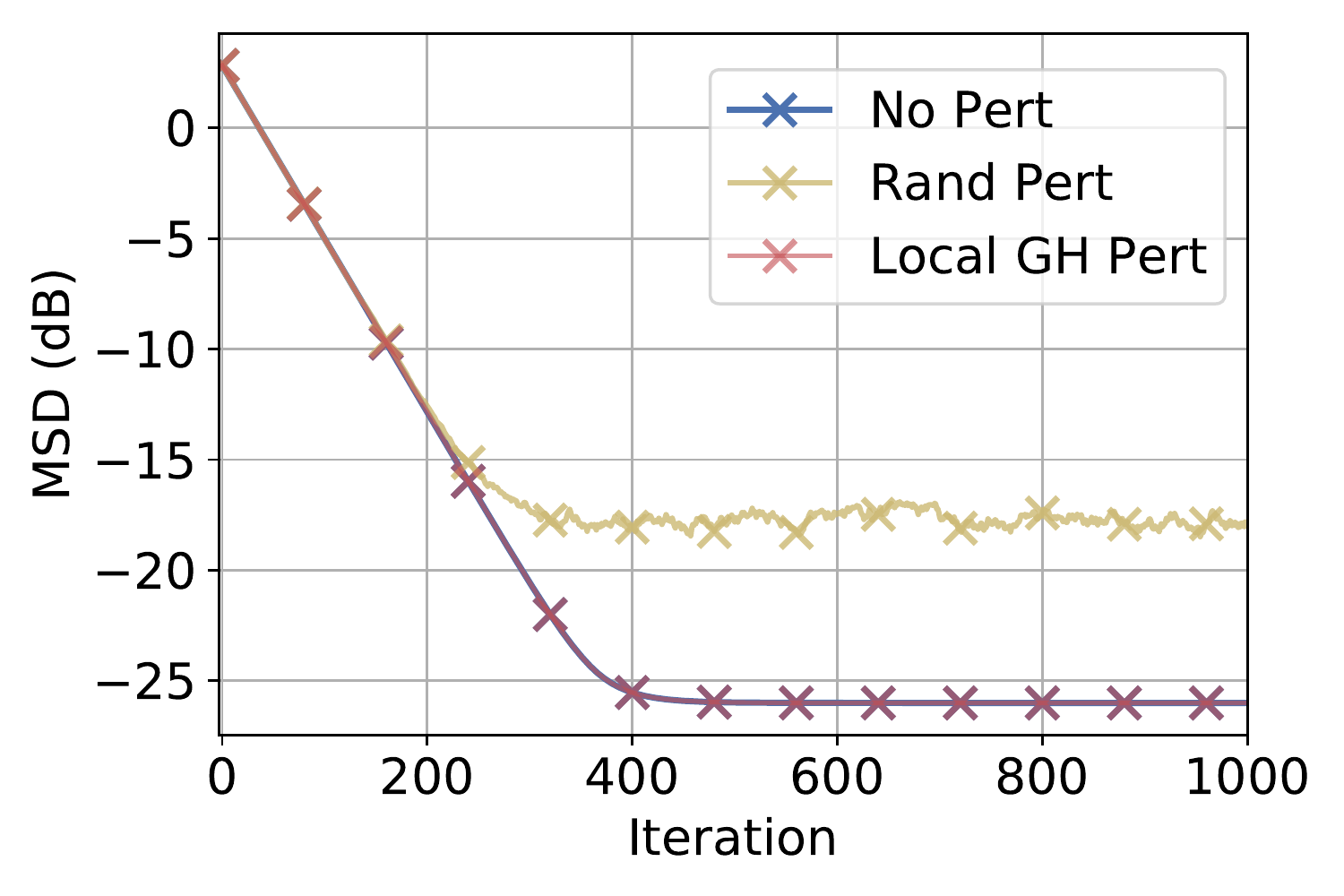}
		\caption{Average individual MSD}
	\end{subfigure}
	\caption{MSD plots for the distributed learning algorithms.}\label{fig:MSD-dist}
\end{figure}

\section{Conclusion}\label{sec:con}
We introduce a distributed random number generator and apply it to a distributed learning setting to ensure differential privacy without degradation in performance.

%


\vfill\pagebreak




\begin{thebibliography}{10}
\providecommand{\url}[1]{#1}
\csname url@samestyle\endcsname
\providecommand{\newblock}{\relax}
\providecommand{\bibinfo}[2]{#2}
\providecommand{\BIBentrySTDinterwordspacing}{\spaceskip=0pt\relax}
\providecommand{\BIBentryALTinterwordstretchfactor}{4}
\providecommand{\BIBentryALTinterwordspacing}{\spaceskip=\fontdimen2\font plus
\BIBentryALTinterwordstretchfactor\fontdimen3\font minus
  \fontdimen4\font\relax}
\providecommand{\BIBforeignlanguage}[2]{{%
\expandafter\ifx\csname l@#1\endcsname\relax
\typeout{** WARNING: IEEEtran.bst: No hyphenation pattern has been}%
\typeout{** loaded for the language `#1'. Using the pattern for}%
\typeout{** the default language instead.}%
\else
\language=\csname l@#1\endcsname
\fi
#2}}
\providecommand{\BIBdecl}{\relax}
\BIBdecl

\bibitem{Apt2009}
K.~R. Apt, E.-R. Olderog, and K.~Apt, ``Distributed programs,'' in
  \emph{Verification of Sequential and Concurrent Programs}.\hskip 1em plus
  0.5em minus 0.4em\relax London: Springer, 2009, pp. 373--406.

\bibitem{verbraeken2020survey}
J.~Verbraeken, M.~Wolting, J.~Katzy, J.~Kloppenburg, T.~Verbelen, and J.~S.
  Rellermeyer, ``A survey on distributed machine learning,'' \emph{ACM
  Computing Surveys}, vol.~53, no.~2, pp. 1--33, 2020.

\bibitem{bonawitz2016practical}
K.~Bonawitz, V.~Ivanov, B.~Kreuter, A.~Marcedone, H.~B. McMahan, S.~Patel,
  D.~Ramage, A.~Segal, and K.~Seth, ``Practical secure aggregation for
  privacy-preserving machine learning,'' in \emph{Proc. ACM SIGSAC Conference
  on Computer and Communications Security}, New York, USA, 2017, pp.
  1175--1191.

\bibitem{Mohassel2017SecureML}
P.~{Mohassel} and Y.~{Zhang}, ``Secureml: A system for scalable
  privacy-preserving machine learning,'' in \emph{IEEE Symposium on Security
  and Privacy (SP)}, San Jose, CA, USA, 2017, pp. 19--38.

\bibitem{froelicher2021scalable}
D.~Froelicher, J.~R. Troncoso-Pastoriza, A.~Pyrgelis, S.~Sav, J.~S. Sousa,
  J.-P. Bossuat, and J.-P. Hubaux, ``Scalable privacy-preserving distributed
  learning,'' \emph{Proceedings on Privacy Enhancing Technologies}, vol. 2021,
  no.~2, pp. 323--347, 2021.

\bibitem{Niko2013}
V.~{Nikolaenko}, U.~{Weinsberg}, S.~{Ioannidis}, M.~{Joye}, D.~{Boneh}, and
  N.~{Taft}, ``Privacy-preserving ridge regression on hundreds of millions of
  records,'' in \emph{IEEE Symposium on Security and Privacy}, Berkeley, CA,
  USA, 2013, pp. 334--348.

\bibitem{dwork2014algorithmic}
C.~Dwork and A.~Roth, ``The algorithmic foundations of differential privacy.''
  \emph{Found. Trends Theor. Comput. Sci.}, vol.~9, no. 3-4, pp. 211--407,
  2014.

\bibitem{JayDLDP}
B.~Jayaraman, L.~Wang, D.~Evans, and Q.~Gu, ``Distributed learning without
  distress: Privacy-preserving empirical risk minimization,'' in \emph{Advances
  in Neural Information Processing Systems}, Montreal, Canad, 2018, p.
  6346–6357.

\bibitem{LiDLDP}
C.~{Li}, P.~{Zhou}, L.~{Xiong}, Q.~{Wang}, and T.~{Wang}, ``Differentially
  private distributed online learning,'' \emph{IEEE Transactions on Knowledge
  and Data Engineering}, vol.~30, no.~8, pp. 1440--1453, 2018.

\bibitem{pathak2010DP}
M.~A. Pathak, S.~Rane, and B.~Raj, ``Multiparty differential privacy via
  aggregation of locally trained classifiers.'' in \emph{Advances in Neural
  Information Processing Systems}, Vancouver, Canada, 2010, pp. 1876--1884.

\bibitem{vlaski2020graphhomomorphic}
S.~Vlaski and A.~H. Sayed, ``Graph-homomorphic perturbations for private
  decentralized learning,'' in \emph{Proc. ICASSP}, Toronto, Canada, June 2021,
  pp. 5240--5244.

\bibitem{GFL}
E.~Rizk and A.~H. Sayed, ``A graph federated architecture with privacy
  preserving learning,'' in \emph{IEEE International Workshop on Signal
  Processing Advances in Wireless Communications}, Lucca, Italy, 2021, pp.
  1--5.

\bibitem{10.1145/3243734.3243756}
S.~N. Cohney, M.~D. Green, and N.~Heninger, ``Practical state recovery attacks
  against legacy rng implementations,'' in \emph{{\em Proc.} ACM SIGSAC
  Conference on Computer and Communications Security}, Toronto, Canada, 2018,
  p. 265–280.

\bibitem{johnston2018random}
D.~Johnston, \emph{Random Number Generators -- Principles and Practices}.\hskip
  1em plus 0.5em minus 0.4em\relax De Gruyter Press, 2018.

\bibitem{8907316}
T.~Nguyen-Van, T.-D. Le, T.~Nguyen-Anh, M.-P. Nguyen-Ho, T.~Nguyen-Van, M.-Q.
  Le-Tran, Q.~N. Le, H.~Pham, and K.~Nguyen-An, ``A system for scalable
  decentralized random number generation,'' in \emph{IEEE International
  Enterprise Distributed Object Computing Workshop}, 2019, pp. 100--103.

\bibitem{cascudo2017scrape}
I.~Cascudo and B.~David, ``Scrape: Scalable randomness attested by public
  entities,'' in \emph{International Conference on Applied Cryptography and
  Network Security}, Kanazawa, Japan, 2017, pp. 537--556.

\bibitem{syta2017scalable}
E.~Syta, P.~Jovanovic, E.~K. Kogias, N.~Gailly, L.~Gasser, I.~Khoffi, M.~J.
  Fischer, and B.~Ford, ``Scalable bias-resistant distributed randomness,'' in
  \emph{IEEE Symposium on Security and Privacy}, San Jose, California, 2017,
  pp. 444--460.

\bibitem{hanke2018dfinity}
T.~Hanke, M.~Movahedi, and D.~Williams, ``Dfinity technology overview series,
  consensus system,'' \emph{arXiv:1805.04548}, 2018.

\bibitem{schindler2018hydrand}
P.~Schindler, A.~Judmayer, N.~Stifter, and E.~Weippl, ``Hydrand: Practical
  continuous distributed randomness,'' \emph{Cryptology ePrint Archive}, 2018.

\bibitem{popov2017decentralized}
S.~Popov, ``On a decentralized trustless pseudo-random number generation
  algorithm,'' \emph{Journal of Mathematical Cryptology}, vol.~11, no.~1, pp.
  37--43, 2017.

\bibitem{Blu}
M.~Blum, ``Coin flipping by telephone a protocol for solving impossible
  problems,'' \emph{SIGACT News}, vol.~15, no.~1, p. 23–27, jan 1983.

\bibitem{diffieHell}
W.~Diffie and M.~Hellman, ``New directions in cryptography,'' \emph{IEEE
  Transactions on Information Theory}, vol.~22, no.~6, pp. 644--654, 1976.

\bibitem{DeGroot}
M.~H. DeGroot, ``Reaching a consensus,'' \emph{Journal of the American
  Statistical Association.}, vol.~69, no. 345, pp. 118--121, 1974.

\bibitem{Johansson2008SubgradientMA}
B.~Johansson, T.~Keviczky, M.~Johansson, and K.~H. Johansson, ``Subgradient
  methods and consensus algorithms for solving convex optimization problems,''
  in \emph{Proc. IEEE Conf. Dec. Control (CDC)}, Cancun, Mexico, December 2008,
  pp. 4185--4190.

\bibitem{Nedic09}
A.~Nedic and A.~Ozdaglar, ``Distributed subgradient methods for multi-agent
  optimization,'' \emph{IEEE Transactions on Automatic Control}, vol.~54,
  no.~1, pp. 48--61, 2009.

\bibitem{chen2012diffusion}
J.~Chen and A.~H. Sayed, ``Diffusion adaptation strategies for distributed
  optimization and learning over networks,'' \emph{IEEE Transactions on Signal
  Processing}, vol.~60, no.~8, pp. 4289--4305, Aug 2012.

\bibitem{Tu12}
S.-Y. Tu and A.~H. Sayed, ``Diffusion strategies outperform consensus
  strategies for distributed estimation over adaptive networks,'' \emph{IEEE
  Transactions on Signal Processing}, vol.~60, no.~12, pp. 6217--6234, Dec
  2012.

\bibitem{sayed2014adaptation}
A.~H. Sayed, ``Adaptation, learning, and optimization over networks,''
  \emph{Foundations and Trends in Machine Learning}, vol.~7, no. 4-5, pp.
  311--801, 2014.

\end{thebibliography}
\end{document}